\documentclass[journal]{IEEEtran}
\usepackage{cite}
\usepackage{amsmath,amssymb,amsfonts}
\usepackage{algorithmic}
\usepackage{graphicx}
\usepackage{textcomp}
\usepackage{amsthm}
\usepackage{algorithm}
\usepackage{color}
\usepackage{soul}
\usepackage{setspace}
\usepackage{multirow}
\usepackage{array}

\newtheorem{theorem}{Theorem} 
\newtheorem{proposition}{Proposition}

\begin{document}

\title{A Polynomial Interpolation based Quantum Key Reconciliation  Protocol: Error Correction without Information Leakage}
\author{Gunes Karabulut Kurt,Senior Member, IEEE,
Enver Ozdemir, Member, IEEE, Neslihan Aysen Ozkirisci, Ozan Alp Topal, Student Member, IEEE and Emel A. Ugurlu

\thanks{G. Karabulut Kurt and O. A. Topal are with Department of Electronics and Communication Engineering at Istanbul Technical University, , Istanbul, Turkey. (e-mail: gkurt@itu.edu.tr; topalo@itu.edu.tr)}
\thanks{E. Ozdemir is with Informatics Institute, Istanbul Technical University, Istanbul, Turkey. (e-mail: ozdemiren@itu.edu.tr)}
\thanks{N. A. Ozkirisci is with Department of Mathematics, Yildiz Technical University, Istanbul, Turkey. (e-mail: aozk@yildiz.edu.tr)}
\thanks{E. A. Ugurlu is with Department of Mathematics, Marmara University, Istanbul, Turkey. (e-mail: emel.aslankarayigit@marmara.edu.tr).}}

\maketitle

\begin{abstract}
In this work, we propose a novel key reconciliation protocol for the quantum key distribution (QKD). Based on Newton's polynomial interpolation, the proposed protocol aims to correct all erroneous bits at the receiver without revealing information to the eavesdropper. We provide the exact frame error rate (FER) expression of the proposed protocol. The inherent nature of the proposed algorithm ensures correcting all erroneous bits if the algorithm succeeds. We present an information-theoretical proof that the revealed information during the key reconciliation process is equal to zero. We also provide a numerical comparison of our algorithm with the asymptotic performance of the error-correcting codes and two exemplary low-density-parity-check (LDPC) codes. The results highlight that our algorithm provides superior performance when compared to the LDPC codes, regardless of the distance between Alice and Bob. Furthermore, the proposed key reconciliation protocol is usable for the longer quantum link distances than the state-of-the-art protocols.
\end{abstract}


\IEEEpeerreviewmaketitle

\section{Introduction}
\label{sec:introduction}
\IEEEPARstart{B}{y} exploiting the fundamental laws of quantum physics, quantum key distribution (QKD) promises a theoretically unbreakable shield for a shared message between two distant nodes, Alice and Bob. As described in the first QKD protocol \cite{bennett}, BB84, the uncertainty principle and the no-cloning principle respectively avoid any eavesdropper to correctly decode the shared message and  to hide from Alice and Bob \cite{nocloning}. Considering these assurances, the shared message between Alice and Bob is assumed to be information-theoretically secure \cite{lo1999unconditional}, and QKD provides the means to share the secret keys that can be utilized for cryptographic purposes. Following BB84, QKD has expanded into an active area of research, both in theoretical and practical aspects. New key agreement protocols have been proposed  \cite{protocol1}, \cite{protocol2}, \cite{protocol3}, and their security proofs have been obtained \cite{lo1999unconditional}.

As shown in Figure \ref{fig:QKD_system}, in a QKD setup, two different links are assumed to be available between Alice and Bob: a quantum link for the quantum key agreement process, and an authenticated public communication link (shown as the public channel) for the key reconciliation process. Both of these links may be intercepted by an unauthorized eavesdropper node, Eve.   Alice generates the secret key message (key generation) and transmits the secret message to Bob via the quantum link. The transmitted message bits are converted into qubits (bit/qubit), where generally are realized by the polarized light beams and transmitted over a fiber network or a free space optical (FSO) network \cite{scarani2009security}. Any eavesdropping activity on the quantum link disrupts changes the state of the transmitted qubit, and result in errors in the transmitted message. Besides eavesdropping activity, the transmitted message signal also fades as the link distance increases, named as path loss. Therefore, Alice and Bob first determine the error rate of their shared key (\textit{key sifting}). Then, they utilize the public authenticated communication link to correct the erroneous message bits at the receiver, Bob (\textit{key reconcilation}). The error correction operation is named as key reconciliation, and any shared information over the public channel is assumed to be perfectly obtained by Eve. Since the shared information is related to the secret key, Eve can obtain some amount of information regarding the shared message. Therefore, Alice and Bob both remove the related information that Eve may capture in this process and consequently resulting in a reduced key length.  This problem is referred  as information leakage and reduces the efficiency of the key reconciliation algorithm. The success rate of a key reconciliation protocol is defined by frame error rate (FER), where it's value is equal to the probability of successful reconciliation.

\begin{figure}[t]
	\centering
	\includegraphics[width=\linewidth]{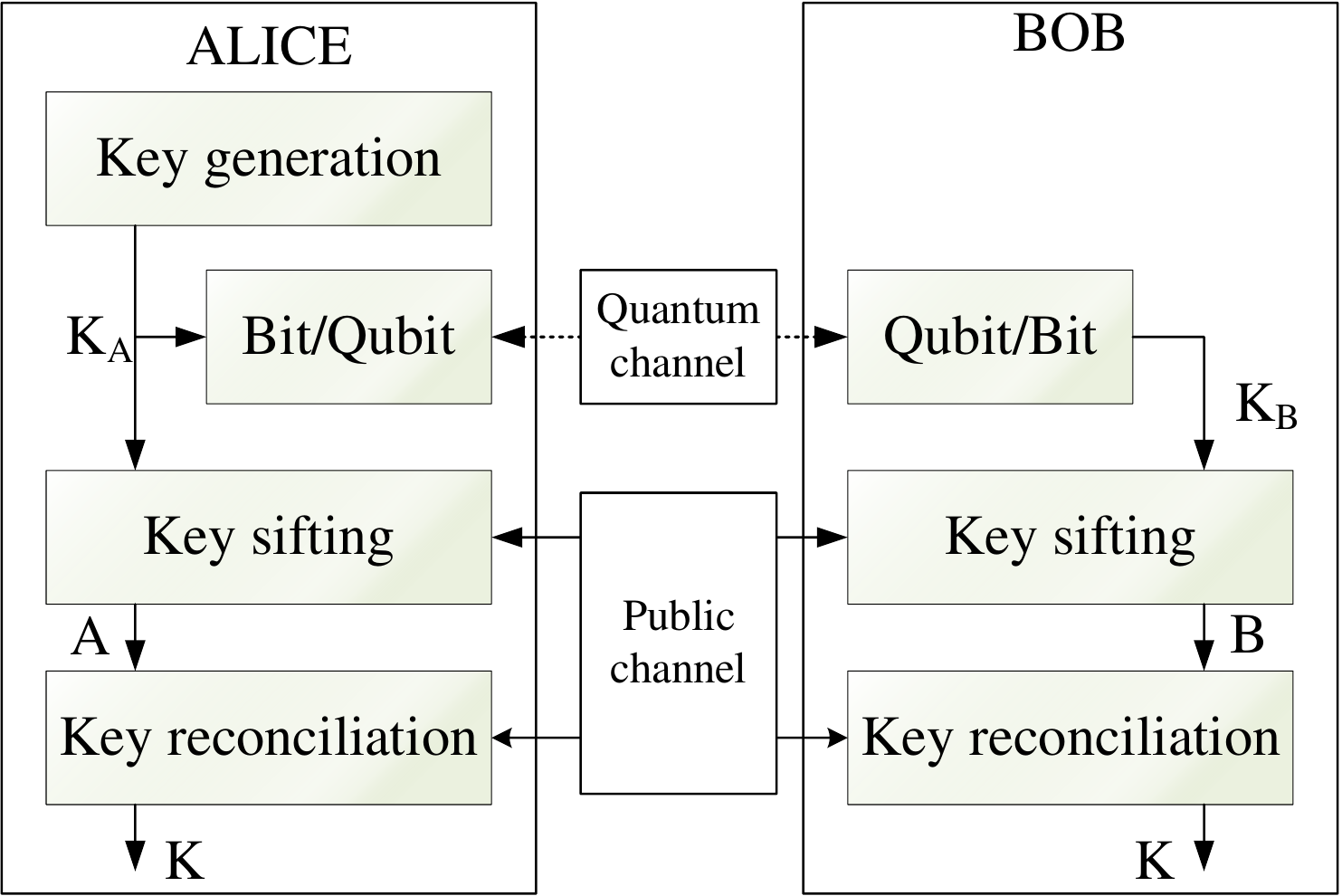}
	\caption{The system model of QKD. The proposed protocol corresponds to key reconciliation block of the system model.}
	\label{fig:QKD_system}
\end{figure}

The first key reconciliation algorithm, Cascade, targets correcting the erroneous bits with minimum information leakage \cite{brassard1993secret}. As the length of the quantum link increases, the key reconciliation becomes more challenging due to the increased error rate. Furthermore, as the error rate increases, the amount of communication between Alice and Bob also increases, which referred to as communication overload \cite{rec-compare}. To meet these requirements, over the decades, some improvements are introduced to the Cascade algorithm, but its main disadvantage, the communication overload, is still an open issue \cite{newcascade}, \cite{ma2010improvement}.

Following Cascade, Winnow protocol is proposed to overcome the communication overload problem. However, by introducing additional errors (error propagation) at the receiver, Winnow protocol limits the length of the QKD link. More prominent key reconciliation is obtained by the utilization of the error-correcting codes (ECC). They are easy to integrate into QKD systems for the key reconciliation purposes since they are already intensively researched for reversing the disruptive channel fading effects in the mobile communication networks. While the computational complexity varies according to the selected code, error propagation and information leakage are still open issues in these code-based approaches \cite{errorcorrection}, \cite{elkouss2010information}. By providing favorable error rate performance and comparatively low complexity structures, low-density-parity-check (LDPC) codes are the most intensely utilized ECC types, and the-state-of-the-art results are given in \cite{ldpc}.   More recently, in our previous work, we have proposed a key verification protocol that finds and removes all of the erroneous bits by blocking any information reveal to Eve in this process \cite{KUR}. One drawback of this protocol is that Alice and Bob both lose some of the correct bits as well as the erroneous bits since we remove them from the shared message completely. The literature still lacks a key reconciliation protocol that can correct all erroneous bits without revealing any information even considering long quantum links, where the initial error rates are very high.

In this paper, we propose a key reconciliation protocol that aims to correct all erroneous bits with zero information leakage. The fundamental idea behind our key reconciliation protocol is based on polynomial interpolation as Shamir's (k,n) thresholding scheme. Before starting the key reconciliation, the original and received messages at respectively Alice and Bob are divided into blocks with equal lengths, and the key reconciliation protocol is applied to each block separately. First, Alice selects a certain number of bits in a block and generates random numbers in a finite field with the same number of selected bits. Then, the protocol interpolates a polynomial, where the vertical axis values of the polynomial are the selected bits, and the horizontal axis values of the polynomial are the randomly generated numbers. The vertical axis values for the remaining bits in the block are obtained according to the interpolated polynomial. Then, all horizontal axis values are shared via the public communication channel. The aim of Bob is obtaining the same polynomial with the received message bits and shared horizontal axis values. Bob can obtain the same polynomial if and only if he has a certain number of correct bits. By interpolating the polynomial and obtaining its roots, Bob can find which part of the message is erroneous and correct them. Also, the shared information can not be used by Eve unless she knows a predetermined number of bit-streams through the key bits, which is limited by the nature of the QKD system \cite{nocloning}.
 
{The main contributions of this paper can be listed as below;
	\begin{itemize}
		\item We propose a polynomial interpolation based key reconciliation protocol for QKD. We obtain the exact FER expression for our algorithm.
		
		\item The proposed reconciliation protocols also ensures identical secret keys at Alice and Bob, if the protocol succeeds. Therefore, the error propagation through the key reconciliation process is prevented by the inherent characteristics of the proposed approach.
		
		\item  We provide an information theoretical proof that the proposed protocol leaks no information to Eve, by solely sharing unrelated information during the key reconciliation process.
		
		\item We provide numerical analysis to compare the performance of our protocol with the asymptotical performance of the error correcting codes and two of the state-of-the-art LDPC codes considering a fiber link and an FSO link for the quantum key agreement part.

\end{itemize}}

One drawback of our algorithm is the increased computational complexity in comparison to state-of-the-art key verification algorithms.

The remainder of this paper is structured as follows. In the following section, we give the mathematical foundations on polynomial interpolation, root finding and describe the physical parameters of the QKD system model. Then, in Section III, the proposed key reconciliation algorithm is described. In Section IV, we provide the performance analysis metrics of our algorithm. The numerical analysis is given in Section V, and finally, the concluding remarks are drawn and the future work is presented in Section VI.

\begin{figure*}[t]
	\begin{center}
		\centering
		\includegraphics[width=0.85\textwidth]{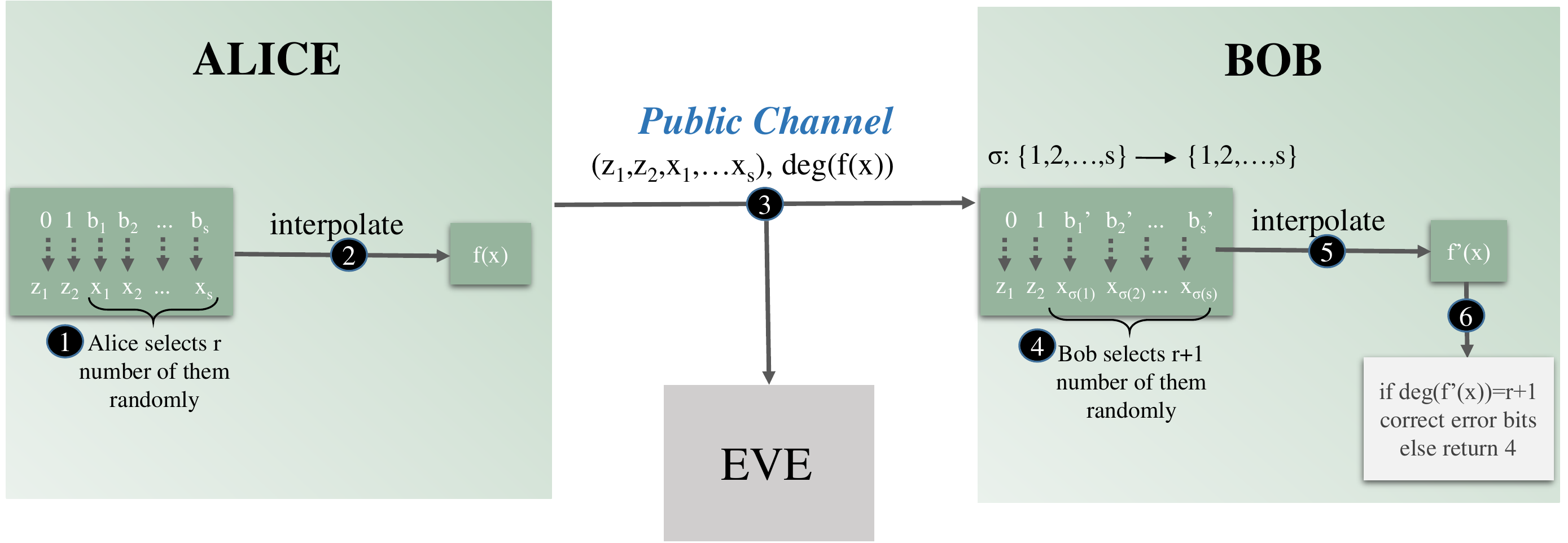}
	\end{center}
	\caption{Depiction of the proposed key reconciliation process. \label{overflow}}
\end{figure*}

\section{Preliminaries}
The polynomial interpolation along with the root finding are the main ingredients of the proposed protocol. Therefore, the first part of this section is devoted to a brief summary of these ingredients over finite fields, and then we provide the basic concepts of the QKD system.
\subsection{Polynomial Interpolation}
Let $(x_0,y_0),\dots,(x_n,y_n)$ be points on Euclidean space, Newton theorem says there exists a unique polynomial of degree at most $n$ interpolating these points.
\begin{theorem}\label{NTTheorem}
{\protect\cite{KIN}}	Let $(x_0,y_0),\dots,(x_k,y_k)$ be points on the graph of a function $f(x)$. There exists a unique polynomial $p(x)$ of degree $\le k$ such that $p(x_i)=f(x_i)=y_i$ for $i=0,\dots,k$.
\end{theorem}

Finding such an interpolating polynomial $p(x)$ is not a tedious task if one applies Newton's divided difference method or Lagrange interpolation. Restricting the degree of a polynomial provides the uniqueness and it is the main factor that polynomial interpolation has been applied to certain problems in secure digital communication \cite{SHA}. The security depends on the hardness of the following statement.

\begin{proposition}\label{Sec}
	Let $p(x)$ be a polynomial of degree $k$ over a finite field with size $\ge 10^{20}$. The probability of a random point on Euclidean plane being on the graph of $p(x)$ is negligible.	
\end{proposition}
\begin{proof}
	See \protect\cite{SHA}
\end{proof}

For example, if the polynomial is over the real field $\mathbb R$, it is even impossible to find a point on the graph of $p(x)$ without knowing the actual polynomial. On the other hand, the polynomials that are going to be employed for our purposes lie in a finite field. The selection of a large size finite field is going to give confidence that a random point will not be on the graph. That is the reason for utilizing polynomials in secret sharing and quantum key exchange. In other words, let $p(x)$ be a polynomial of degree $k$ over the real numbers $\mathbb R$, it is impossible to construct $p(x)$ even if $k$ points on the graph of it is known. 

\subsection{Root Findings} \label{RootFinding}
Let $\mathbb F_p$ be a finite field with $p$ number of elements where $p$ is a prime integer. Let $f(x)$ be a polynomial of degree $k$ defined over $\mathbb F_p$. The proposed algorithm will seek roots of $f(x)$ over the field $\mathbb F_p$ or any extension of it. The polynomial $f(x)$ is constructed randomly, therefore it is hard to predict where its roots lie. On the other hand, several mathematical tools allow one to locate an extension field of $\mathbb F_p$ where the roots of $f(x)$ belong.  Note that, the $\gcd(a,b)$ denotes the greatest common divider of $a$ and $b$.
\begin{theorem} \label{FT}
	Let $f(x)$ be polynomial of degree $k$ in $\mathbb F_p[x]$. The degree of $\gcd(x^p-x,f(x))$ gives the number of roots of $f(x)$ which lie on $\mathbb F_p$.
\end{theorem}

\begin{proof}
	This is an easy consequence of Fermat's little theorem which basically says $a\in \mathbb F_p$ if and only if $a^p\equiv a \mod p$. 
\end{proof}  
Theorem \ref{FT} suggests a method to decide where the other roots lie. For example, the degree of $\gcd(x^{p^2}-x,g(x))$ gives the number roots in the degree 2 extension field of $\mathbb F_p$. Similarly, computing $\gcd(x^{p^e}-x,f(x))$ for $e=2,\dots,k$ gives the places of all roots of $f(x)$. This process is in general called distinct degree factorization, see [Section 1.6]\cite{COH} for more details. Once the location of roots is established, root-finding algorithms can be applied. Consider the distinct degree factorization gives an extension degree of $\mathbb F_p$ where the roots are located. For example, if some roots are in $\mathbb F_{p^{e_1}}$ and the rest is in $\mathbb F_{p^{e_2}}$ then root-finding algorithms should be performed for the field $\mathbb F_{p^{\text{lcm}(e_1,e_2)}}$ where $\text{lcm}(e_1,e_2)$ is the least common multiple of $e_1$ and $e_2$. \\
\indent Once the extension field is determined, the next aim to find all roots. Let $f(x)$ be the polynomial over $\mathbb F_p$ such that all of its roots lie in an extension $\mathbb F_q$ where $q=p^e$. There are various algorithms to find a root of $f(x)$ in $\mathbb F_q$. For example, a root-finding algorithm defined in [Section 1.6]\cite{COH} describes a probabilistic method to determine at least one root of $f(x)$. Basically, it suggests randomly selecting a member $\alpha \in  \mathbb F_q$ and computing $\gcd(x^{\frac{q-1}{2}}-\alpha,f(x))$. This might result in a root of $f(x)$. On the other hand, the proposed algorithm needs to compute more than one root, this method would not be effective for our purposes. Therefore, at one point our method is going to employ a polynomial factorization algorithm. \\
\indent The practical polynomial factorization algorithms are all probabilistic \cite{BER, CAN-ZESS, OZD}. Let $f(x)$ be a reducible polynomial over $\mathbb F_p$ where $p$ is a prime number.  Berlekamp's algorithm and Cantor-Zassenhaus algorithm first searches a polynomial $h(x)$ such that $$h(x)^p\equiv h(x) \mod p$$ Then, the algorithms decides an element $c\in \mathbb F_p$ and computes $\gcd(f(x),h(x)-c)$. The probability of success for both algorithms is around $1/2$. On the other hand, the algorithm described in \cite{OZD} works on any extension field and the probability of success is always more than $3/4$. The algorithm employs singular curves and their Jacobian groups.

\subsection{Quantum Key Distribution}
\textcolor{black}{In this part, we provide the basics of QKD. For the following of this part, we adopted the notation and equations in \cite{martinez2013key}.}
As shown in Figure \ref{fig:QKD_system}, a QKD system consists of two links: $(i)$ a quantum channel for the key agreement process and $(ii)$ an authenticated public communication channel for the key distillation process. In the key agreement process, Alice transforms the generated key bit string $K_A$ of length $2L$ into qubits and shares them via the quantum channel. Bob measures these qubits and maps measurement results into his key bit string $K_B$ of length $2L$. Due to the characteristics of the quantum channel and possible eavesdropping activity, Bob does not know the accuracy of his measurements and needs to estimate the disparities between the $K_A$ and $K_B$. Key distillation is a post processing step that is used to obtain identical secret keys at Alice and Bob. The process starts with key sifting.  In key sifting, Alice and Bob share half of the randomly selected bits and estimate quantum bit error rate (QBER) denoted by $\epsilon$. After key distillation, if QBER is lower than the error toleration rate $\Gamma$, Alice and Bob remove the shared bits, and each of them starts key reconciliation with $L$ number of bits. The remaining bit string at Alice is denoted by $A$, and the remaining bit sitring at Bob is denoted by $B$. The proposed key reconciliation process is defined by $R^{\epsilon}(A,B)=[\mathcal{S},Y]$, where the final secret key string $\mathcal{S}$ is obtained from the correlated bit strings $A$ and $B$ by exchanging information string $Y$.

QBER can be evaluated by $$ \epsilon=\frac{p_d}{p_{exp}},$$ where $p_d$ is the dark count rate. $p_{exp}$ is the total photon detection rate, and $$p_{exp}=p_{signal}+p_d-p_{signal}p_d.$$ $p_{sig}$ denotes the signal detection rate, and can be approximated by $$p_{signal}=\mu t \eta,$$ where $\mu$, $t$ and $\eta$ denote respectively the quantum efficiency of the detector, transmitivity and the average number of emitted photons per pulse. In order to obtain secret key rate in the following sections, an upper bound for the fraction of detected single photons, $\Upsilon_1$, can be given by $$\hat{\Upsilon_1}= 1- \frac{p_{multi}(\mu)}{p_{exp}} \leq \Upsilon_1.$$
$p_{multi}(\mu)=1-(1+\mu)e^{-\mu}$ denotes the probability of emitting two or more photons. The error rate of single photon pulses is upper-bounded by $\varepsilon=\frac{\epsilon}{\hat{\Upsilon_1}}$.

$\mathcal{H}(A)$ denotes the Shannon entropy of the binary random variable $A\in\{a_1,a_2\}$, where the probabilities $P(A=a_1)=\theta$ and $P(A=a_2)=1-\theta$. $\mathcal{H}(\theta)$ can be given by
$$\mathcal{H}(\theta)=-\theta\log_2(\theta)-(1-\theta)\log(1-\theta).$$
\textcolor{black}{
	By using the provided equalities, we can estimate the QBER performance of any designed QKD system by simply adapting the physical parameters of the test system.}
\section{The Proposed Key Reconciliation Protocol}
As described in Section II.C., Alice and Bob start the key reconciliation with $L$ number of bits. As a first step Alice and Bob divide their respectively keys $A$ and $B$ into $m$ number of blocks with length of $s$, where $L=ms$. The key reconciliation operation is the same for all blocks so let us assume that both Alice and Bob work on the first block. That is Alice's side block is $$B_A=01b_1b_2\dots b_s$$ and Bob's side block is $$B_B=01b_1'b_2'\dots b'_s.$$

Suppose the error toleration rate is $\Gamma$ and  $$r=\lfloor s-s\cdot {\Gamma}\rfloor$$ In literature, $\Gamma$ can be at most 0.15 \cite{rec-compare} and in our case we assume $\Gamma$ is less than 0.30.  Alice randomly selects $r+2$ number of distinct elements $z_1,z_2, x_1,x_2,\dots,x_r \in \mathbb F_p$  and sets up the pairs:
$$(z_1,0),(z_2,1),(x_1,b_1), (x_2,b_2),\dots,(x_{r},b_r)$$

Then she constructs $f(x)$ of degree less than $r+2$ interpolating the above points. In the next step, she solves:
$$\begin{array}{cccc}
f(x)&=&b_{r+1}\\
f(x)&=&b_{r+2}\\
\cdot & \cdot & \cdot\\
\cdot & \cdot & \cdot\\
\cdot & \cdot & \cdot\\
f(x) &=&b_{s}
\end{array}$$
for $x$. Suppose she obtains $x_{r+1},x_{r+2},\dots,x_{s}$. Note that $x$ coordinates must be distinct.
Alice then broadcasts $z_1,z_2,x_1,x_2,\dots,x_{s}$. The preselection of first  two bits 0 and 1 prevents the polynomial $f(x)$ being a constant polynomial in case all $b_i$ are the same. Note that the probability of degree $f(x)$ being different than $r+1$ is negligible by Proposition 1. Therefore, we might assume degree of $f(x)$ is $r+1$ from now on.  If $s-(r+1)$ is larger than the $r+1$ and $b_{r+1}=b_{r+2}=\dots=b_{s}$ then the roots $x_{r+3},\dots,x_{s}$ would not be exist. That is the reason, we assume the number $\Gamma$ is always less than 30.  

\begin{algorithm}[tb]  
	\caption{: Key Established Algorithm: Alice's Side}           
	\begin{algorithmic}[1]   
		\renewcommand{\algorithmicrequire}{\textbf{Input:}}
		\renewcommand{\algorithmicensure}{\textbf{Output:}}
		\REQUIRE A block of bit string $B_A$ with size $s$.
		\ENSURE $(z_1,z_2,x_1,x_2,\dots, x_{s})\in \mathbb F_p^{s+2}$ and the degree $r+1$.
		\STATE Select random elements $z_1,z_2\in \mathbb F_p$ and a random permutation  function $$\delta:\{1,2,\dots,s\}\rightarrow \{1,2,\dots,s\}$$
		Then again randomly select $x_{\delta(1)},x_{\delta(2)},\dots,x_{\delta(r)}\in\mathbb F_p$.
		
		\STATE Construct pairs:
		$$ (z_1,0),(z_2,1),(x_{\delta(1)},b_{\delta(1)}), (x_{\delta(2)},b_{\delta(2)}),\dots, (x_{\delta(r)},b_{\delta(r)})$$
		\STATE Find a polynomial $f(x)\in \mathbb F_p[x]$ interpolating above set of pairs.
		\STATE Set 
		$$\begin{array}{cccc}
		f(x)&=& b_{\delta(r+1)}\\
		f(x)&=&b_{\delta(r+2)}\\
		\cdot &\cdot &\cdot\\
		\cdot &\cdot &\cdot\\
		\cdot &\cdot &\cdot\\
		f(x)&=&b_{\delta(s)}
		\end{array}$$
		and find $x_{\delta(r+1)},\dots,x_{\delta(s)}$.
		\RETURN The set $(z_1,z_2,x_1,\dots,x_s)$ and the degree $r+1$.	
	\end{algorithmic}
	\label{AliceAlg} 
\end{algorithm}

\indent Bob receives the sequence $z_1,z_2,x_1,x_2,\dots, x_{s}$ from Alice in an open network. Bob's side constructed $B_B$ via quantum channel. He  selects  a random permutation function $$\sigma:\{1,2,\dots,s\}\rightarrow \{1,2,\dots,s\}$$  and then constructs set of pairs $$ (z_1,0),(z_2,1),(x_{\sigma(1)},b'_{\sigma(1)}),\dots, (x_{\sigma(r)},b'_{\sigma(r+1)})$$
Bob then finds a polynomial $f'(x)$ interpolating all these pairs. Once the degree of $f'(x)$ is $r+1$, then Bob concludes his polynomial is the same as Alice's one. Note that Bob interpolates $r+3$ pairs and if they are not on the graph of Alice's polynomial, the probability that Bob's polynomial is of degree $r+1$ is negligible by Proposition 1. Constructing the same function with Alice allows Bob to remove error terms in the bit string in his side. Suppose his $j^{\text{th}}$ bit doesn't match with the $j^{\text{th}}$ bit in Alice's side, that is $b_j\ne b'_j$. He just observes this discrepancy while computing $f(x_j)$.

\begin{algorithm}[tb]  
	
	\caption{: Key Established Algorithm: Bob's Side}           
	\begin{algorithmic}[1]   
		\renewcommand{\algorithmicrequire}{\textbf{Input:}}
		\renewcommand{\algorithmicensure}{\textbf{Output:}}
		\REQUIRE A block of bit string $B_B$ and $(z_1,z_2,x_1,\dots,x_{s}) \in\mathbb F_p^{s+2}$ and the degree $ r+1. $
		\ENSURE Compare $B_B$ to $B_A$ and correct error bits.
		\STATE Select random permutation function $$ \sigma:\{1,2,\dots, s\} \rightarrow  \{1,2,\dots, s\} $$ and $r$ elements in $(x_{\sigma{(1)}},\dots,x_{\sigma{(s)}})$
		
		\STATE Construct pairs $$ (z_1,0),(z_2,1),(x_{\sigma(1)},b'_{\sigma(1)}),\dots, (x_{\sigma(r+1)},b'_{\sigma(r+1)})$$ and obtain a polynomial $ f'(x)\in  \mathbb F_p[x]$ interpolating these pairs.
		\IF {$ \deg f'(x)==r+1  $ } 
		
		\FOR {$ i=r+2 \text{ to } s$}
		\IF {$ f'(x_{\sigma(i)})== b'_{\sigma(i)}$}
		\STATE keep $ b'_{\sigma(i)}$
		\ELSE
		\STATE replace $ b'_{\sigma(i)}$  with $ f'(x_{\sigma(i)}). $
		\ENDIF
		\ENDFOR
		
		\ELSE 
		\STATE return Step 1.
		
		\ENDIF
		\RETURN The bit string $B_A$.
	\end{algorithmic}
	\label{BobAlg}             
\end{algorithm}

The key is established at Alice's side using  Algorithm 1. The integer $s$ represents the number of elements in each block.  First, Alice selects two random elements $z_1,z_2\in \mathbb F_p$ and a random permutation map $\delta$ on $\{1,2,\dots,s\}$. Then, she also randomly selects $x_{\delta(1)},x_{\delta(2)},\dots,x_{\delta(r)}\in\mathbb F_p$ (Step 1). Following that, she sets up $r+2$ pairs $$ (z_1,0),(z_2,1),(x_{\delta(1)},b_{\delta(1)}), (x_{\delta(2)},b_{\delta(2)}),\dots, (x_{\delta(r)},b_{\delta(r)})$$ (Step 2). In the next step, she constructs the interpolation polynomial $f(x)\in \mathbb F_p[x]$ with degree $r+1$ (Step 3). By using root findings methods, she determines the remaining pairs $$ (x_{\delta(r+1)},b_{\delta(r+1)}), (x_{\delta(2)},b_{\delta(2)}),\dots, (x_{\delta(s)},b_{\delta(s)})$$ (Step 4). At the end, the algorithm outputs the set $(z_1,z_2,x_1,\dots,x_s)$ and the degree $r+1$.

The key established method of Bob's side is given in Algorithm 2. Initially, Bob selects a random permutation map $\sigma$ on $\{1,2,\dots,s\}.$ He also selects $r$ elements in $x_{\sigma{(1)}},\dots,x_{\sigma{(s)}}$(Step 1). Upon completion of previous step, he sets up $r+3$ pairs $$ (z_1,0),(z_2,1),(x_{\sigma(1)},b'_{\sigma(1)}),\dots, (x_{\sigma(r+1)},b'_{\sigma(r+1)})$$ and finds the interpolation polynomial $ f'(x)\in  \mathbb F_p[x]$ (Step 2). Since Bob's aim is to find the Alice's polynomial, he checks whether $\deg f'(x)$ is $r+1$. He keeps changing the selection of $r+1$ pairs in  $(x_{\sigma(1)},b'_{\sigma(1)}),\dots, (x_{\sigma(s)},b'_{\sigma(s)})$ until he gets  $\deg(f'(x))=r+1$. In the remaining part of Algorithm 2, he checks the equation $ f'(x_{\sigma(i)})= b'_{\sigma(i)}$ for $ i=r+2 ,\dots, s.$  If the equation satisfies, he concludes that the bit $b'_{\sigma(i)}$ matches with Alice's side. In the other case, he changes the bit $b'_{\sigma(i)}$ with $ f'(x_{\sigma(i)}).$ At the end of the process, Bob finds  wrong bits in his side and corrects them. 
\section{Performance Analysis}
\subsection{Computation Complexity}
\subsubsection{Complexity of Algorithm \ref{AliceAlg}}
\indent Algorithm \ref{AliceAlg} first randomly selects element in finite field $\mathbb F_p$. Then pairs up $r$ of these randomly selected elements with the generated bits. Then interpolates $r+2$ pairs and constructs a polynomial $f(x)$ of degree $r+1$. The next steps requires to find roots of $f(x)$ over an extension field $\mathbb F_q$ of $\mathbb F_p$. In terms of running time, finding the roots of a polynomial of degree $r+1$ is much more costly than interpolating $r+2$ pairs. Therefore, with respect to running time, the dominating step of the algorithm is the root finding step. Finding roots of a polynomial $f(x)$ over $\mathbb F_p$ requires several sub-steps as mentioned in the Subsection \ref{RootFinding}. The first sub-step is to locate the extension field which contains the roots of $f(x)$. This requires performing several greatest common divisor algorithm for polynomials. There exist many algorithms to compute the greatest common divisor between two polynomials. In the respect of the cost of computation, calculating $ \gcd $ of two polynomials of degree at most $ r $ in $\mathbb F_p$ takes $\mathcal{O}((2r^2+r)\log p)$ operations with the algorithm given in [Section 6.9] \cite{GAT} and \cite{CF}. We prefer to apply a polynomial factorization algorithm to find roots of $f(x)$ over the extension field $\mathbb F_q$ where the all roots lie and $q=p^e$ for some integer $e>0$. Since most of the time all roots will not be in $\mathbb F_p$ instead they will be an extension field of it, the algorithm defined in \cite{OZD} is more convenient to be employed. Therefore, the polynomial factorization algorithm to find all roots will costs at most $\mathcal O(r^4\log p)$. This is because the algorithm's running time is approximately $\mathcal O(r^3\log q)$ and $q\le p^r$. Overall, the complexity of Algorithm \ref{AliceAlg} is $\mathcal O(r^4\log p)$ operations.
\subsubsection{Complexity of Algorithm \ref{BobAlg}}
The algorithm running on the receiver's side is similar to the algorithm presented in \cite{KUR} which costs $\mathcal O(r^2)$ operations.

\subsection{Frame Error Rate}
In error correction coding, FER is equal to zero if the codeword is correct, and equal to one even if one-bit mismatches in the codeword. Similarly, in QKD, FER is only equal to zero if all erroneous bits are correct after key reconciliation. \textcolor{black}{However considering the variety of error patterns and random error distribution, FER becomes the expected value of a random variable, where the outcome of the variable is equal to 0 for a successful reconciliation and 1 for a fail.  As described in Section III, Bob reconciles a block of the secret key by finding the degree of the polynomial, consequently, the polynomial itself, if he can correctly decode at least $r+1$ number of bits in a block of $s$ number of bits. In this case, considering a single block, FER is equal to the probability of having at least $r+1$ number of correct bits in the block of Bob's key. Considering $k$ consecutive blocks, FER for the key reconciliation protocol becomes the $k^{\text{th}}$ power of the aforementioned probability. Then}, assuming quantum channel as a binary symmetric channel with error probability $\epsilon=\mathrm{QBER}$ \cite{martinez2013key}, \cite{scarani2009security}, the exact FER for our algorithm can be given by 
\begin{equation}
\mathrm{FER}=1-\left(\sum_{z=r+1}^{s}(1-\epsilon)^{z}(\epsilon)^{l-z}\right)^k,
\end{equation}
\textcolor{black}{
	where $\epsilon$ can be evaluated as in Subsection II.C.
	One observation from this equation is that FER would increase at the same QBER as we divide the same key into more blocks. This property provides a trade-off between FER and complexity performance. Besides, a trade-off between FER and the security level can be observed from the selection of the $r$. As we select  larger $r$ values, the protocol becomes more prominent to eavesdropping attacks, while FER performance reduces since Bob requires to find more correct bits at his block.} FER is one of the main indicators for the key reconciliation process. In the following section, we will provide a comparison of our algorithm and error-correcting codes in terms of $\mathrm{FER}$. 

\subsection{Analysis of the Information Leakage}
As depicted in Figure \ref{fig:QKD_system}, during the QKD process, the information regarding to key can be disclosed with two different information exchange processes: key agreement part in the quantum channel and key reconciliation part in the standard communication  channel  \cite{scarani2009security}. The total amount of disclosed information to Eve is referred as information leakage and denoted by
$$
I_{lk}\triangleq I(A;E),
$$
where $I(A;E)$ denotes the mutual information of the secret at  Alice and  Eve. Then, we can state that
$$I_{lk}(\mathcal{S})=I^{q}_{lk}(\mathcal{S})+I^{p}_{lk}(\mathcal{S}).$$ Here $I^{q}_{lk}(\mathcal{S})$ is the amount of leaked information  during the quantum key agreement part where Alice and Bob use the quantum communication link. $I^{p}_{lk}(\mathcal{S})$ is the amount of leaked information during the key reconciliation part where Alice and Bob use the classical public link.

The revealed information rate at quantum key agreement part can be given by \cite{martinez2013key}
$$I_{lk}^{q}(\mathcal{S})=\Upsilon_1\mathcal{H}(\varepsilon)L.$$ 

During the key reconciliation, the transmission medium is assumed to be an authenticated public channel, where any transmitted information during the key reconciliation process is assumed to be obtained from Eve. In the existing key reconciliation schemes, Alice and Bob aim correcting erroneous bits at Bob by sharing side information about the secret key. According to Slepian-Wolf limit \cite{Wolf}, revealed information for the key reconciliation processes with error control coding codes is lower bounded by 
\begin{equation}
\mathcal{H}(\epsilon)\leq\rho_{\text{leak}}.
\end{equation}

By intuition, Alice reveals the minimum amount of information if she knows Bob's information about the secret. This assures information leakage during the key reconciliation process as long as $\epsilon>0$. More specifically, the revealed information in an $(n,k)$ error correcting code is $I_{lk}^{p}(\mathcal{S})=n-k$, and the revealed information rate is described  by
\begin{equation}
\rho_{\text{leak}}= \frac{n-k}{n}.
\end{equation}

On the other hand, we block the information leakage during the key reconciliation process. By lower bounding the degree of the polynomial with the revealed information to Eve in quantum channel, Eve cannot interpolate the polynomial with her limited vertical axis observations and horizontal axis values shared in key reconciliation process.
\begin{theorem}
	Revealed information for our proposed key reconciliation process is zero for $\Gamma\leq\frac{r-1}{n}$.
\end{theorem}
\begin{proof}
	
	During the proposed key reconciliation scheme, the only shared information is the $y$ axis values of the generated polynomial,  $\mathbf{y}$. Let us denote the mutual information of $\mathbf{y}$ and the secret with $I(\mathcal{S};Y)$. Then, $${I}_{lk}(\mathcal{S})\triangleq I(\mathcal{S};Y)= H(\mathcal{S})-H(\mathcal{S}|Y).$$ For the selected $\Gamma\leq\frac{r-1}{n}$, Eve is assumed to obtained at most $(r-1)$ consecutive bits from the Alice's key during the key agreement process. Let us rewrite the uncertainty about the secret at Eve as $H(\mathcal{S}|Y,A_{r-1}),$ where $A_{r-1}$ denotes any permutation of $(r-1)$ number of bits from $A$ and note that from the chain rule $$H(\mathcal{S}|Y,A_{r-1})\leq H(\mathcal{S}|Y).$$ As proved in Shamir's (k,n) thresholding scheme \cite{SHA}, $r-1$ revealed element pairs $(b_i, x_i)$, does not reveal any information about polynomial of degree $r$, where every candidate secret $S$ corresponds to a unique polynomial of degree $r-1$. From the construction of the polynomials, all their probabilities are equal. Thus, if $$H(\mathcal{S}|Y,A_{r-1})=H(\mathcal{S}|Y)=H(\mathcal{S})$$ then, $I(\mathcal{S};Y)=0.$
\end{proof}
\textcolor{black}{
	Reducing information leakage to zero indicates that we can utilize all secret key bits in the final key insead of removing the leaked part. Therefore, secret key throughput performance improves considering all distances and initial QBER values as detailed in the following part.}

\subsection{Secret Key Throughput}
Similar to the goodput analysis in communication networks, secret key throughput indicates the number of available secret key bits generated in a second. Secret key throughput is defined as
\begin{equation}
S=r_sf_{sc},
\label{eq:S}
\end{equation}
where $r_s$ is the secret key rate. $f_{sc}$ denotes the generated secret key bits in a second, more formally referred to as the frequency of the key source. The secret key rate is given by
\begin{equation}
r_s=(1-\mathrm{FER})p_{exp}q\rho, 
\label{eq:s}
\end{equation}
where $\rho$ denotes the secret key fraction with one-way reconciliation schemes and FER is defined in (1). Since we only consider the successfully reconciled secret keys at the end of the QKD process, the rate of the unsuccessful reconciliation processes (in other words FER) is subtracted by the secret key fraction with one-way reconciliation schemes that obtained after the quantum key agreement part. The secret key fraction in our case can be given by
\begin{equation}
\rho=\Upsilon_1(1-\mathcal{H}(\varepsilon))-\rho_{\text{leak}}.
\label{eq:rho}
\end{equation}

(\ref{eq:S})-(\ref{eq:rho}) highlight the effect of $\mathrm{FER}$ and information leakage in key reconciliation to the secrecy throughput of the QKD system. 
\section{Numerical Analysis}
\label{sc:numerical}
In this section, we present the parameters for numerical analyses and the results in parts. First, we present the numerical analyses regarding the quantum key exchange part of the QKD and then we will provide the performance analyses of different key reconciliation protocols considering different quantum link realizations.
\subsection{Quantum Key Agreement}
In numerical analysis, we mainly consider two different implementations of the quantum link: a fiber link and an FSO link. The parameters for these links are given in Table \ref{table:parameters} and are detailed in Section II. Note that, the parameters are directly obtained from the experimental works of QKD. 

\begin{table}[tb]
	\centering
	\caption{The link realization parameters for numerical analysis. The parameters for fiber link and FSO link are obtained from respectively \protect\cite{martinez2013key} and \protect\cite{kurtsiefer2002step}. }
	\begin{tabular}{|l|l|l|}
		\hline
		\textbf{Realization Parameters} & \textbf{Fiber}     & \textbf{FSO}       \\ \hline
		\multicolumn{1}{|l|}{$f_s$ [GHz]} & $1$    & $1$    \\ \hline
		\multicolumn{1}{|l|}{$K$ [km]}  & $0-60$    & $0-30$  \\ \hline
		\multicolumn{1}{|l|}{$q$}                      & $0.5$     & $0.5$     \\ \hline
		\multicolumn{1}{|l|}{$p_d$}                    & $10^{-5}$ & $10^{-5}$ \\ \hline
		\multicolumn{1}{|l|}{$\alpha$  {[}dB/km{]}}     & 0.2      & 0.1      \\ \hline
		\multicolumn{1}{|l|}{$\eta$}                   & 0.1       & 0.15      \\ \hline
		\multicolumn{1}{|l|}{$d_s$ [cm] }               & n.a.      & 25       \\ \hline
		\multicolumn{1}{|l|}{$d_r$ [cm]}               & n.a.      & 25        \\ \hline
		\multicolumn{1}{|l|}{$D$}                      & n.a.      & 3.21    \\ \hline
	\end{tabular}
	
	\label{table:parameters}
\end{table}

The transmitivity  of the fiber link is given by
\begin{equation}
t^{\text{fiber}}=10^{\frac{-\alpha K}{10}},
\end{equation}
where $K$ is the distance in km \cite{martinez2013key}. The transmitivity of the FSO link is given by
\begin{equation}
t^{\text{FSO}}= {\left(\frac{d_r}{d_s+Dl}\right)}^2 10^{\frac{-\alpha K}{10}},
\end{equation}
where $d_s$ and $d_r$ are the apartures of the transmitting and receiving telescopes. $D$ is the divergence of the beam. Note that, $d_s$, $d_r$ and $D$ are not applicable for the fiber link. QBER values with corresponding distances for the fiber and the FSO links are given in Figure \ref{fig:QBER}. Even though the QBER performance of FSO link is better than fiber in  short distance regime (0-3 km), QBER increases faster in FSO link in comparison with fiber link. Therefore, the key reconciliation becomes more challenging especially in the long distance FSO links. 
\begin{figure}[t]
	\centering
	\includegraphics[width=\linewidth]{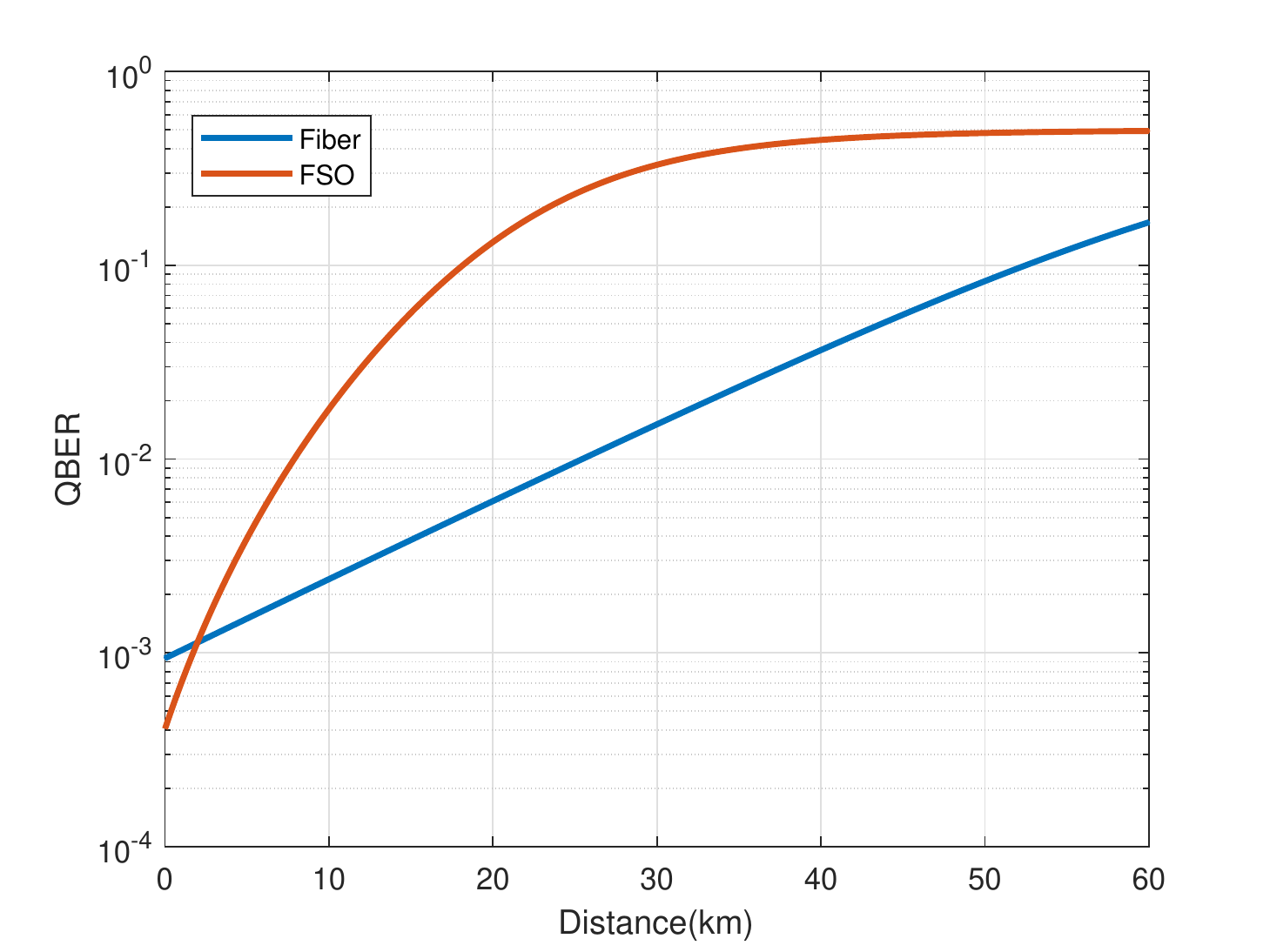}
	\caption{The relationship between QBER and link distance for fiber and FSO links.}
	\label{fig:QBER}
\end{figure}
\subsection{Key Reconciliation}
We consider two different LDPC codes and an asymptotic limit for the error correcting codes to compare the performance of our algorithm. LDPC-1 and LDPC-2 codes in this paper directly correspond to respectively (1008, 504) regular LDPC Code-3 and (1998, 1776) regular LDPC Code-4 in \cite{LDPC_FER}. We also evaluate the $\mathrm{FER}$ values for the corresponding codes from the Eq. (5) of \cite{LDPC_FER} with the parameters as in Table 1 of \cite{LDPC_FER}. The leaked information for LDPC codes are obtained by Eq. (3). We also consider the asymptotic performance for the ECC based key reconciliation schemes. By setting $\rho_{\text{leak}}$ as in Eq. (2) for the asymptotic ECC case, we assume that the leaked information equals to the Slepian-Wolf limit. For our algorithm, we consider $s=100$, $m=10$ and $r=30$.
\begin{figure}[tb]
	\centering
	\includegraphics[width=\linewidth]{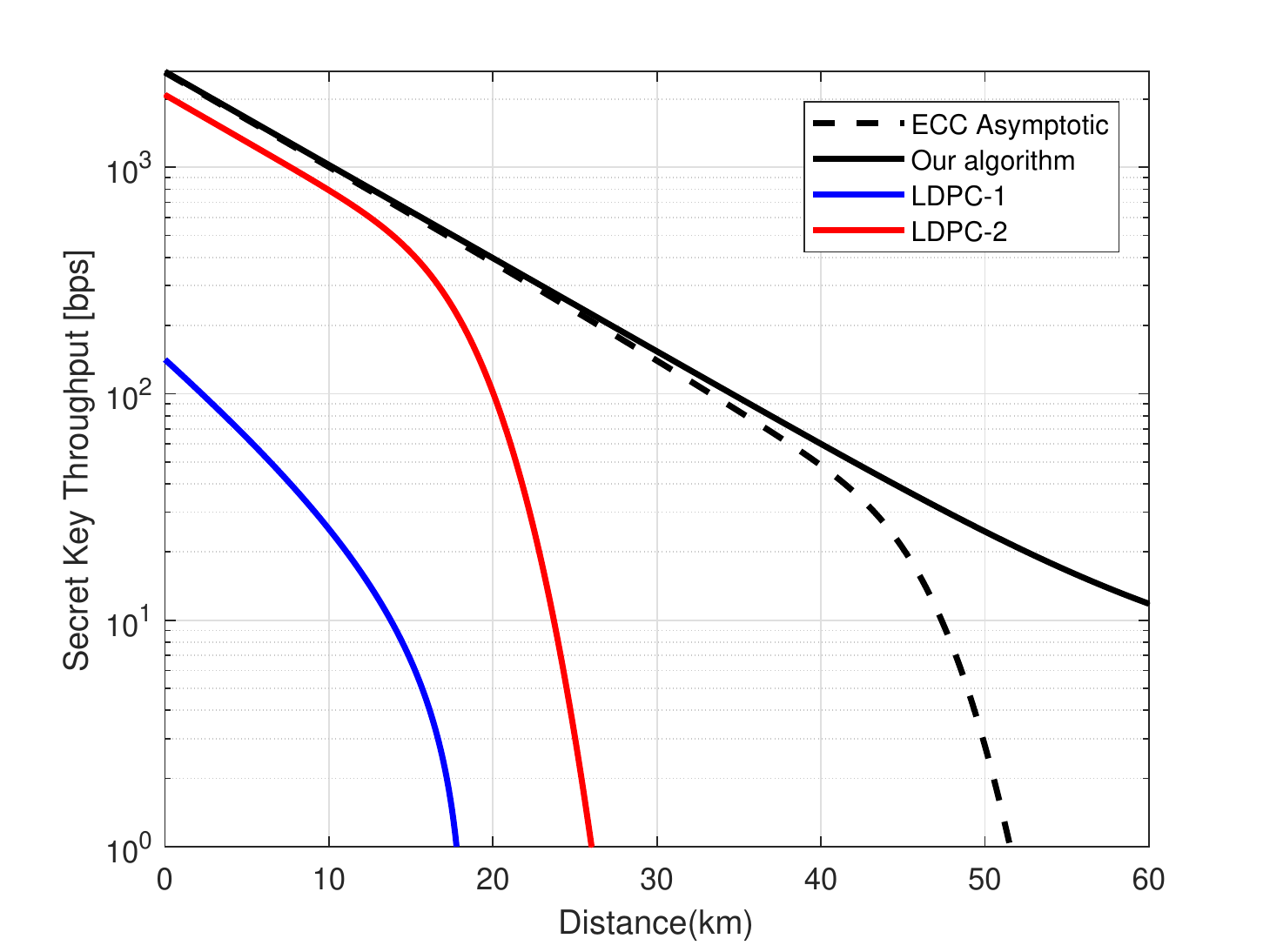}
	\caption{Secret key throughput comparison of different key reconciliation algorithms considering fiber links. 
	}
	\label{fig:fiber}
\end{figure}
Figure \ref{fig:fiber} shows the change in secret key throughputs with link distance considering 2 different LDPC codes, our algorithm and asymptotic performance of ECC. As it can be observed from the figure, the proposed algorithm outperforms other coding schemes and their asymptotic limit. Since our algorithm provides zero information leakage in the key reconciliation process, it outperforms other methods in the short distance region. In other words, the proposed algorithm provides the same secret key throughput at longer link distances than the ECC based key reconciliation schemes. For example, the proposed algorithm provides 100 bps secret key throughput at 33 km, while LDPC-1 and LDPC-2 provide the same throughput value at respectively  2 and 20 km distances.  Furthermore, our algorithm can also perform in the long-distance regime contrary to ECC schemes, since it has very low FER rates even in the high QBER regime.

\begin{figure}[tb]
	\centering
	\includegraphics[width=\linewidth]{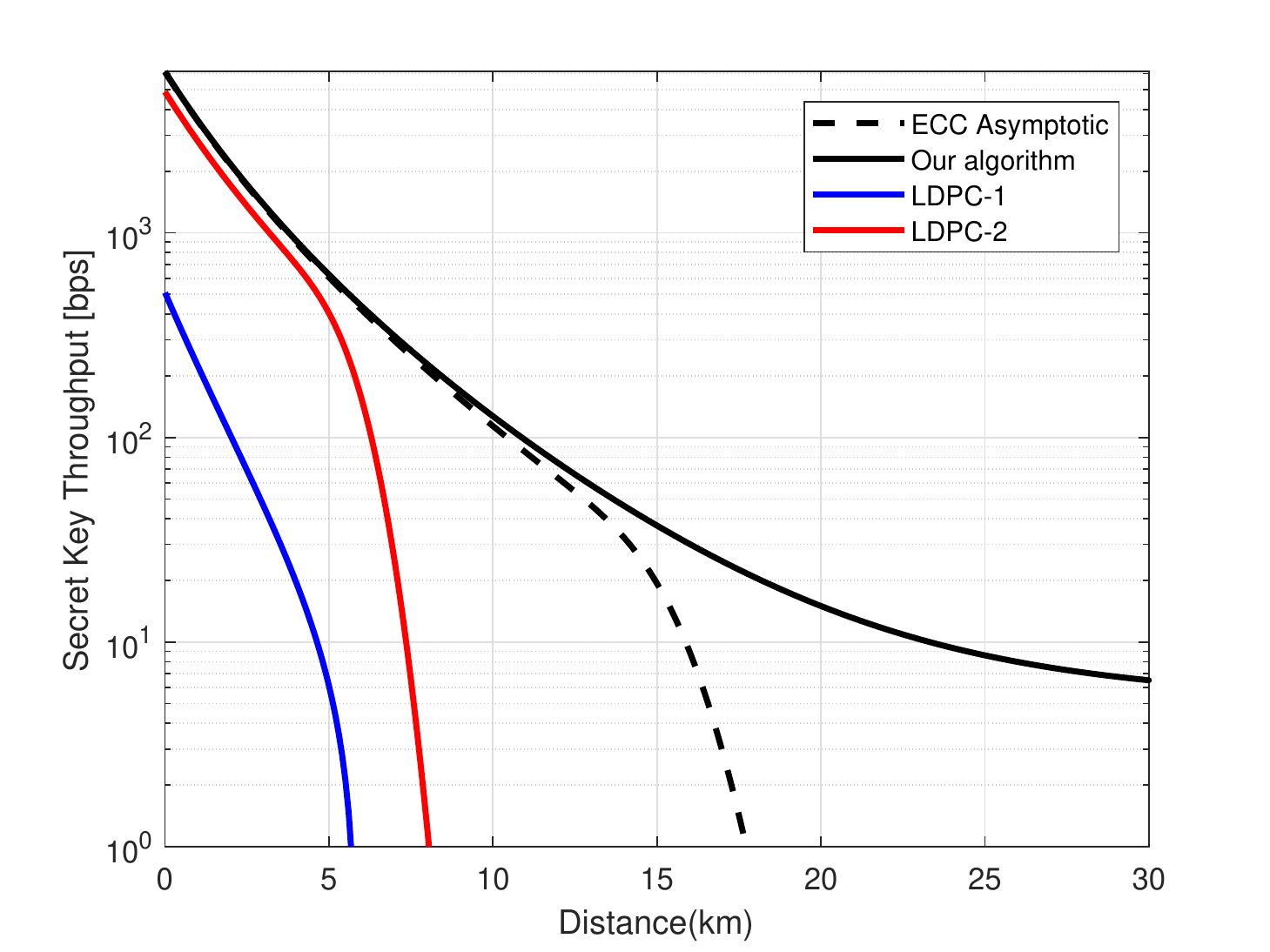}
	\caption{Secret key throughput comparison of different key reconciliation algorithms considering FSO links.
	}
	\label{fig:FSO}
\end{figure}
Figure \ref{fig:FSO} illustrates the secret key throughputs considering the FSO quantum link. Since QBER increases faster in the FSO link, the FER increases more rapidly as the distance increases. Therefore, the performance of ECC schemes drops faster than the fiber link. Our algorithm can also perform at longer distances in FSO link than ECC schemes. Since QBER becomes 0.5 after 20 km in FSO, the detection in the binary symmetric channel almost becomes arbitrary. Naturally, the error floor occurs after 20 km distance in our algorithm.

\textcolor{black}{
	Considering both of the figures, the proposed algorithm provides a higher secret key throughput at all link distances. Especially for the short distance interval, the improved performance results from the zero information leakage property of the proposed algorithm. As the distance between Alice and Bob increases, FER of our algorithm is not affected as much as error-correcting codes. Therefore, our algorithm can pave the way to the implementation of longer QKD links.}

\section{Conclusion}
In this paper, we proposed a polynomial interpolation based key reconciliation protocol for quantum key distribution (QKD). By solely sharing unrelated information with the secret key, the revealed information to the eavesdropper is proven to be zero during the key reconciliation process. We provided the exact frame error rate (FER) expression for the proposed protocol. As illustrated in the numerical studies, the proposed protocol outperforms error-correcting codes by the means of secret key throughput considering fiber and FSO links. We consider reducing the computational complexity in the future work by optimizing the block length and the number of blocks.

\bibliographystyle{IEEEtran}
\bibliography{IEEEexample}

\end{document}